\documentclass[a4paper,UKenglish]{lipics-v2018}

\hideLIPIcs
\usepackage{microtype}

\usepackage{amsthm}

\newcommand{\augment}{\textsc{Augment}}
\newcommand{\ignore}[1]{}
\newcommand{\K}{\mathbb{K}}
\newcommand{\BigO}{\mathcal{O}}
\newcommand{\eps}{\varepsilon}
\newcommand{\reals}{\mathbb{R}}
\newcommand{\BigOT}{\tilde{\mathcal{O}}}
\newcommand{\opt}{\textsc{Opt}}
\newcommand{\runningtime}{\BigOT( m(\sqrt{w}+ \sqrt{r}+\frac{wr}{n}))}
\newcommand{\dist}{{\sf c}}

\newcommand{\grid}{\mathbb{G}}
\newcommand{\ymax}{y_{max}}
\newcommand{\etal}{\textit{et al}.}

\newcommand{\R}{\mathcal{R}}

\newtheorem*{lemma*}{Lemma}
\newtheorem*{corollary*}{Corollary}

\nolinenumbers
\title{A Weighted Approach to the Maximum Cardinality Bipartite Matching Problem with Applications in Geometric Settings}

\titlerunning{Weighted Approach to Maximum Cardinality Bipartite Matching}

\author{Nathaniel Lahn}{Department of Computer Science, Virginia Tech, USA}{lahnn@vt.edu}{}{}

\author{Sharath Raghvendra}{Department of Computer Science, Virginia Tech, USA}{sharathr@vt.edu}{}{}

\authorrunning{N. Lahn and S. Raghvendra}

\subjclass{ Theory of computation $\rightarrow$  Design and analysis of algorithms $\rightarrow$ Graph algorithms analysis $\rightarrow$  Network flows}

\Copyright{Nathaniel Lahn and Sharath Raghvendra}
\begin{document}

\maketitle
\begin{abstract}

We present a weighted approach to compute a maximum cardinality matching in an arbitrary bipartite graph.
Our main result is a new algorithm that takes as input a weighted bipartite graph $G(A\cup B,E)$ with edge weights of $0$ or $1$. Let $w \leq n$ be an upper bound on the weight of any matching in $G$. Consider the subgraph induced by all the edges of $G$ with a weight $0$. Suppose every connected component in this subgraph has $\BigO(r)$ vertices and $\BigO(mr/n)$ edges. We present an algorithm to compute a maximum cardinality matching in $G$ in $\runningtime$ time.\footnote{We use $\BigOT$ to suppress poly-logarithmic terms.} 

When all the edge weights are $1$ (symmetrically when all weights are $0$), our algorithm will be identical to the well-known Hopcroft-Karp (HK) algorithm, which runs in $\BigO(m\sqrt{n})$ time. However, if  we can carefully assign weights of $0$ and $1$ on its edges such that both $w$ and $r$ are sub-linear in $n$  and $wr=\BigO(n^{\gamma})$ for $\gamma < 3/2$, then we can compute maximum cardinality matching in $G$ in $o(m\sqrt{n})$ time.  Using our algorithm, we obtain a new $\BigOT(n^{4/3}/\eps^4)$ time algorithm to compute an $\eps$-approximate bottleneck matching of $A,B\subset\reals^2$ and an $\frac{1}{\eps^{\BigO(d)}}n^{1+\frac{d-1}{2d-1}}\mathrm{poly}\log n$ time algorithm for computing $\eps$-approximate bottleneck matching in $d$-dimensions. All previous algorithms take $\Omega(n^{3/2})$ time. Given any graph $G(A \cup B,E)$ that has an easily computable balanced vertex separator for every subgraph $G'(V',E')$ of size $|V'|^{\delta}$, for $\delta\in [1/2,1)$, we can apply our algorithm to compute a maximum matching in $\BigOT(mn^{\frac{\delta}{1+\delta}})$ time improving upon the $\BigO(m\sqrt{n})$ time taken by the HK-Algorithm. 
 \end{abstract}

\section{Introduction} 
We consider the classical matching problem in an arbitrary unweighted bipartite graph $G(A \cup B, E)$ with $|A|=|B|=n$ and $E\subseteq A\times B$. A \emph{matching} $M\subseteq E$ is a set of vertex-disjoint edges. We refer to a largest cardinality matching $M$ in $G$ as a \emph{maximum matching}. A maximum matching is \emph{perfect} if $|M|=n$. Now suppose the graph is weighted and every edge $(a,b)\in E$ has a \emph{weight} specified by $\dist(a,b)$. The weight of any subset of edges $E'\subseteq E$ is given by $\sum_{(a,b)\in E}\dist(a,b)$. A \emph{minimum-weight maximum matching} is a maximum matching with the smallest weight. In this paper, we present an algorithm to compute a maximum matching faster by carefully assigning weights of $0$ and $1$ to the edges of $G$.

\subparagraph*{Maximum matching in graphs:} In an arbitrary bipartite graph with $n$ vertices and $m$ edges, Ford and Fulkerson's algorithm~\cite{ford_fulkerson} iteratively computes, in each phase, an augmenting path in $\BigO(m)$ time, leading to a maximum cardinality matching in $\BigO(mn)$ time. Hopcroft and Karp's algorithm (HK-Algorithm)~\cite{hk_sicomp73} reduces the number of phases from $n$ to $\BigO(\sqrt{n})$ by computing a maximal set of vertex-disjoint shortest augmenting paths in each phase. A single phase can be implemented in $\BigO(m)$ time leading to an overall execution time of $\BigO(m\sqrt{n})$.
In weighted bipartite graphs with $n$ vertices and $m$ edges, the well-known Hungarian method computes a minimum-weight maximum matching in  $\BigO(mn)$ time~\cite{hungarian_56}.  Gabow and Tarjan designed a weight-scaling algorithm (GT-Algorithm) to compute a minimum-weight perfect matching in $\BigO(m\sqrt{n}\log (nC))$ time, provided all edge weights are integers bounded by $C$~\cite{gt_sjc89}. Their method, like the Hopcroft-Karp algorithm, computes a maximal set of vertex-disjoint shortest (for an appropriately defined augmenting path cost) augmenting paths in each phase. For the maximum matching problem in arbitrary graphs (not necessarily bipartite), a weighted approach has been applied to achieve a simple $\BigO(m\sqrt{n})$ time algorithm~\cite{gabow2017weighted}. 

Recently Lahn and Raghvendra~\cite{lr_soda19}  gave $\BigOT(n^{6/5})$ and  $\BigOT(n^{7/5})$ time algorithms for finding a minimum-weight perfect bipartite matching in planar and $K_h$-minor\footnote{They assume $h=O(1)$.} free graphs respectively, overcoming the $\Omega(m\sqrt{n})$ barrier; see also Asathulla~\etal~\cite{soda-18}. Both these algorithms are based on the existence of an $r$-clustering which, for a parameter $r >0$, is a partitioning of $G$ into edge-disjoint clusters $\{\R_1,\ldots,\R_k\}$ such that $k = \BigOT(n/\sqrt{r})$, every cluster $\R_j$ has $\BigO(r)$ vertices, and each cluster has $\BigOT(\sqrt{r})$ \emph{boundary} vertices. A  boundary vertex has edges from two or more clusters incident on it. Furthermore, the total number of boundary vertices, counted with multiplicity, is $\BigOT(n/\sqrt{r})$. The algorithm of Lahn and Raghvendra extends to any graph that admits an $r$-clustering. 
There are also algebraic approaches for the design of fast algorithms for bipartite matching; see for instance~\cite{madry2013navigating,fast_matrix_matching}. 

\subparagraph*{Matching in geometric settings:}  In geometric settings, $A$ and $B$ are points in a fixed $d$-dimensional space and $G$ is a complete bipartite graph on $A$ and $B$. For a fixed integer $p\ge 1$, the weight of an edge between $a \in A$ and $b\in B$ is $\|a-b\|^p$, where $\|a-b\|$ denotes the Euclidean distance between $a$ and $b$. The weight of a matching $M$ is given by $\bigl(\sum_{(a,b)\in M} \|a-b\|^p\bigr)^{1/p}$. For any fixed $p \ge 1$, we wish to compute a perfect matching with the minimum weight. When $p=1$, the problem is the well-studied \emph{Euclidean bipartite matching} problem. A minimum-weight perfect matching for $p = \infty$ will minimize the largest-weight edge in the matching and is referred to as a \emph{bottleneck matching}. The Euclidean bipartite matching in a plane can be computed in $\BigOT(n^{3/2+\delta})$~\cite{s_socg13} time for an arbitrary small $\delta>0$; see also Sharathkumar and Agarwal~\cite{sa_soda12}. Efrat~\etal\ present an algorithm to compute a bottleneck matching in the plane in $\BigOT(n^{3/2})$~\cite{efrat_algo} time. Both these algorithms use geometric data structures in a non-trivial fashion to speed up classical graph algorithms.

When $p=1$, for any $0<\eps\le 1$,  there is an $\eps$-approximation algorithm for the Euclidean bipartite matching problem that runs in  $\BigOT(n/\eps^d)$ time~\cite{sa_stoc12}. However, for $p >1$, all known $\eps$-approximation algorithms take $\Omega(n^{3/2}/\eps^d)$ time. We note that it is possible to find a $\Theta(1)$-approximate bottleneck matching in $2$-dimensional space by reducing the problem to finding maximum flow in a planar graph and then finding the flow using an $\BigOT(n)$ time max-flow algorithm~\cite{multiple_planar_maxflow}. There are numerous other results; see also~\cite{argawal_transportation_17,sa_stoc_14,argawal_phillips_rms}. Designing exact and approximation algorithms that break the $\Omega(n^{3/2})$ barrier remains an important research challenge in computational geometry. 

\subparagraph*{Our results:} We present a weighted approach to compute a maximum cardinality matching in an arbitrary bipartite graph.
Our main result is a new matching algorithm that takes as input a weighted bipartite graph $G(A\cup B, E)$ with every edge having a weight of $0$ or $1$. Let $w \leq n$ be an upper bound on the weight of any matching in $G$. Consider the subgraph induced by all the edges of $G$ with a weight $0$. Let $\{K_1,K_2,\ldots, K_l\}$ be the connected components in this subgraph and let, for any $1\le i \le l$,  $V_i$ and $E_i$ be the vertices and edges of $K_i$. We refer to each connected component $K_i$ as a \emph{piece}. Suppose $|V_i|=\BigO(r)$ and $|E_i|=\BigO(mr/n)$. Given $G$, we present an algorithm to compute a maximum matching in $G$ in $\runningtime$ time. Consider any graph in which removal of sub-linear number of ``separator'' vertices partitions the graph into connected components with $\BigO(r)$ vertices and $\BigO(mr/n)$ edges. We can apply our algorithm to any such graph by simply setting the weight of every edge incident on any separator vertex to $1$ and weights of all other edges to $0$. 

When all the edge weights are $1$ or all edge weights are $0$, our algorithm will be identical to the HK-Algorithm algorithm and runs in $\BigO(m\sqrt{n})$ time. However, if  we can carefully assign weights of $0$ and $1$ on the edges such that both $w$ and $r$ are sub-linear in $n$  and for some constant $\gamma < 3/2$, $wr=\BigO(n^{\gamma})$, then we can compute a maximum matching in $G$ in $o(m\sqrt{n})$ time.  
Using our algorithm, we obtain the following result for bottleneck matching:
\begin{itemize}
\item Given two point sets $A, B \subset \reals^2$ and an $0< \eps \le 1$, we reduce the problem of computing an $\eps$-approximate bottleneck matching to computing a maximum cardinality matching in a subgraph $\mathcal{G}$ of the complete bipartite graph on $A$ and $B$. We can, in $\BigO(n)$ time assign $0/1$ weights to the  $\BigO(n^2)$ edges of $\mathcal{G}$ with so that any matching has a weight of $\BigO(n^{2/3})$. Despite possibly $\Theta(n^2)$ edges in $\mathcal{G}$, we present an efficient implementation of our graph algorithm with $\BigOT(n^{4/3}/\eps^4)$ execution time that computes an $\eps$-approximate bottleneck matching for $d=2$; all previously known algorithms take $\Omega(n^{3/2})$ time. Our algorithm, for any fixed $d \ge 2$ dimensional space, computes an $\eps$-approximate bottleneck matching in $\frac{1}{\eps^{\BigO(d)}}n^{1+\frac{d-1}{2d-1}}\mathrm{poly}\log n$ time. (See Section~\ref{sec:bottleneck}).
\end{itemize}
The algorithm of Lahn and Raghvendra~\cite{lr_soda19} for $K_h$-minor free graphs requires the clusters  to have a small number of boundary vertices, which is used to create a compact representation of the residual network. This compact representation becomes prohibitively large as the number of boundary vertices increase. For instance, their algorithm has an execution time of $\Omega(m\sqrt{n})$ for the case where $G$ has a balanced vertex separator of $\Theta(n^{2/3})$.  Our algorithm, on the other hand, extends to any graph with a sub-linear vertex separator.  Given any graph $G(A \cup B,E)$ that has an easily computable balanced vertex separator for every subgraph $G'(V',E')$ of size $|V'|^{\delta}$, for $\delta\in [1/2,1)$, there is a $0/1$ weight assignment on edges of the graph  so that the weight of any matching is $\BigO(n^{\frac{2\delta}{1+\delta}})$ and $r= \BigO(n^{\frac{1}{1+\delta}})$. This assignment can be obtained by simply recursively sub-dividing the graph using balanced separators until each piece has $\BigO(r)$ vertices and $\BigO(mr/n)$ edges. All edges incident on the separator vertices are then assigned a weight of $1$ and all other edges are assigned a weight of $0$. As a result, we obtain an algorithm that computes the maximum cardinality matching in $\BigOT(mn^{\frac{\delta}{1+\delta}})$ time. 

\subparagraph*{Our approach:} Initially, we compute, in $\BigO(m\sqrt{r})$ time, a maximum matching within all pieces. Similar to the GT-Algorithm, the rest of our algorithm is based on a primal-dual method and executes in phases.  Each phase consists of two stages. The first stage conducts a Hungarian search and finds at least one augmenting path containing only zero slack (with respect to the dual constraints) edges.  Let the admissible graph be the subgraph induced by the set of all zero slack edges. Unlike in the GT-Algorithm, the second stage of our algorithm computes augmenting paths in the admissible graph that are not necessarily vertex-disjoint.  In the second stage, the algorithm iteratively initiates a DFS from every free vertex. When a DFS finds an augmenting path $P$, the algorithm will augment the matching immediately and terminate this DFS. Let all pieces of the graph that contain the edges of $P$ be \emph{affected}. Unlike the GT-Algorithm, which deletes all edges visited by the DFS, our algorithm deletes only those edges that were visited by the DFS and did not belong to an affected piece. Consequently, we allow for visited edges from an affected piece to be reused in another augmenting path. As a result, our algorithm computes several more augmenting paths per phase than the GT-Algorithm, leading to a reduction of number of phases from $\BigO(\sqrt{n})$ to $\BigO(\sqrt{w})$. Note, however, that the edges of an affected piece may now be visited multiple times by different DFS searches within the same phase. This increases the cumulative time taken by all the DFS searches in the second stage. However, we are able to bound the total number of affected pieces across all phases of the algorithm by $\BigO(w\log w)$. Since each piece has $\BigO(mr/n)$ edges, the total time spent revisiting these edges is bounded by $\BigO(mrw\log (w)/n)$. The total execution time can therefore be bounded by $\runningtime$.

\section{Preliminaries}
We are given a bipartite graph $G(A\cup B, E)$, where any edge $(a,b) \in E$  has a weight $\dist(a, b)$ of $0$ or $1$. Given a matching $M$, a vertex is \emph{free} if it is not matched in $M$. An \emph{alternating path} (resp. cycle) is a simple path (resp. cycle) that alternates between edges in $M$ and not in $M$. An \emph{augmenting path} is an alternating path that begins and ends at a free vertex.

A matching $M$ and an assignment of dual weights $y(\cdot)$ on the vertices of $G$ is \emph{feasible} if for any $(a,b) \in A\times B$:
\begin{align}
y(b) - y(a) \le \dist(a,b) \qquad \textnormal{if } (a, b)\not\in M,\label{eq:feas1}\\
y(a) - y(b) = \dist(a,b) \qquad \textnormal{if } (a, b) \in M. \label{eq:feas2}
\end{align}

  To assist in describing our algorithm, we first define a residual network and an augmented residual network with respect to a feasible matching $M, y(\cdot)$.
A \emph{residual network} $G_M$ with respect to a feasible matching $M$ is a directed graph where every edge $(a, b)$ is directed from $b$ to $a$ if $(a, b) \not\in M$  and from $a$ to $b$ if $(a,b) \in M$. The weight $s(a,b)$ of any edge is given by the slack of this edge with respect to feasibility conditions~\eqref{eq:feas1} and~\eqref{eq:feas2}, i.e., if $(a,b) \not\in M$, then $s(a ,b)= \dist(a,b)+y(a)-y(b)$ and $s(a,b)=0$ otherwise.  An \emph{augmented residual network} is obtained by adding to the residual network an additional vertex $s$ and additional directed edges from $s$ to every vertex in $B_F$, each of having a weight of $0$. We denote the augmented residual network as $G_M'$.

\section{Our algorithm}
\label{sec:graphmatch}
Throughout this section we will use $M$ to denote the current matching maintained by the algorithm and $A_F$ and $B_F$ to denote the vertices of $A$ and $B$ that are free with respect to $M$. Initially $M=\emptyset$, $A_F=A$, and $B_F=B$. Our algorithm consists of two steps. The first step, which we refer to as the \emph{preprocessing step}, will execute the Hopcroft-Karp algorithm and compute a maximum matching within every piece. Any maximum matching $M_{\opt}$ has at most $w$ edges with a weight of $1$ and the remaining edges have a weight of $0$. Therefore, $|M_{\opt}|-|M| \le w$. The time taken by the preprocessing step for $K_i$ is $\BigO(|E_i|\sqrt{|V_i|}) = \BigO(|E_i|\sqrt{r})$. Since the pieces are vertex disjoint, the total time taken across all pieces is $\BigO(m\sqrt{r})$. After this step, no augmenting path with respect to $M$ is completely contained within a single piece. We set the dual weight $y(v)$ of every vertex $v \in A\cup B$ to $0$. The matching $M$ along with the dual weights $y(\cdot)$ satisfies~\eqref{eq:feas1} and~\eqref{eq:feas2} and is feasible.

 The second step of the algorithm is executed in \emph{phases}. We describe phase $k$ of the algorithm. This phase consists of two stages.  
 
 \subparagraph*{First stage:} In the first stage, we construct the augmented residual network $G_M'$ and execute Dijkstra's algorithm with $s$ as the source. Let $\ell_v$ for any vertex $v$ denote the shortest path distance from $s$ to $v$ in $G_M'$. If a vertex $v$ is not reachable from $s$, we set $\ell_v$ to $\infty$. Let
 \begin{equation}
 \label{eq:closestfreevertex}\ell = \min_{v \in A_F}\ell_v.
 \end{equation}
 Suppose $M$ is a perfect matching or $\ell=\infty$, then this algorithm returns with $M$ as a maximum matching. Otherwise, we update the dual weight of any vertex $v \in A\cup B$ as follows. If $\ell_v \ge \ell$, we leave its dual weight unchanged. Otherwise, if $\ell_v < \ell$, we set $y(v) \leftarrow y(v) + \ell - \ell_v$. 
 After updating the dual weights,
 we construct the \emph{admissible graph} which consists of a subset of edges in the residual network $G_M$ that have zero slack. After the first stage, the matching $M$ and the updated dual weights are feasible. Furthermore, there is at least one augmenting path in the admissible graph. This completes the first stage of the phase.
 \subparagraph*{Second stage:} In the second stage, we initialize $G'$ to be the admissible graph and execute DFS to identify augmenting paths. For any augmenting path $P$ found during the DFS, we refer to the pieces that contain its edges as \emph{affected pieces} of $P$.
 
 Similar to the HK-Algorithm, the second stage of this phase will initiate a DFS from every free vertex $b \in B_F$ in $G'$. If the DFS does not lead to an augmenting path, we delete all edges that were visited by the DFS. On the other hand, if the DFS finds an augmenting path $P$, then the matching is augmented along $P$,  all edges that are visited by the DFS and do not lie in an affected piece of $P$ are deleted, and the DFS initiated at $b$ will terminate.  
 
 Now, we describe in detail the DFS initiated for a free vertex $b \in B_F$. Initially $P=\langle b=v_1\rangle$. Every edge of $G'$ is marked unvisited. At any point during the execution of DFS, the algorithm maintains a simple path $P = \langle b=v_1, v_2,\ldots, v_k\rangle$. The DFS search continues from the last vertex of this path as follows: 
 \begin{itemize}
     \item If there are no unvisited edges that are going out of $v_k$ in $G'$, 
     \begin{itemize}
         \item If $P=\langle v_1\rangle$,  remove all edges that were marked as visited from $G'$ and terminate the execution of DFS initiated at $b$.
         \item Otherwise, delete $v_k$ from $P$ and continue the DFS search from $v_{k-1}$,
     \end{itemize} 
     \item If there is an unvisited edge going out of $v_k$, let $(v_k,v)$ be this edge. Mark $(v_k,v)$ as visited. If $v$ is on the path $P$, continue the DFS from $v_k$. If $v$ is not on the path $P$, add $(v_k,v)$ to $P$, set $v_{k+1}$ to $v$, and,
     \begin{itemize}
    \item Suppose $v \in A_F$, then  $P$ is an augmenting path from $b$ to $v$. Execute the \augment\ procedure which augments $M$ along $P$. Delete from $G'$ every visited edge that does not belong to any affected piece of $P$ and terminate the execution of DFS initiated at $b$. 
     \item Otherwise, $v \in (A\cup B)\setminus A_F$. Continue the DFS from $v_{k+1}$.
      \end{itemize}
 \end{itemize}

 The \augment\ procedure receives a feasible matching $M$, a set of dual weights $y(\cdot)$, and an augmenting path $P$ as input. For any $(b,a) \in P\setminus M$, where $a \in A$ and $b \in B$,  set $y(b) \leftarrow y(b)-2\dist(a,b)$. Then augment $M$ along $P$ by setting $M\leftarrow M\oplus P$.    By doing so,  every edge of $M$ after augmentation satisfies the feasibility condition~\eqref{eq:feas2}.
 This completes the description of our algorithm. The algorithm maintains the following invariants during its execution:
 \begin{enumerate}
     \item[(I1)] The matching $M$ and the set of dual weights $y(\cdot)$ are feasible. Let $y_{\max}=\max_{v \in B} y(v)$. The dual weight of every vertex $v \in B_F$ is $y_{\max}$ and the dual weight for every vertex $v \in A_F$ is $0$.
     \item[(I2)]  For every phase that is fully executed prior to obtaining a maximum matching, at least one augmenting path is found and the dual weight of every free vertex of $B_F$  increases by at least $1$.
 \end{enumerate}
 
 \subparagraph*{Comparison with the GT-Algorithm:} In the GT-Algorithm, the admissible graph does not have any alternating cycles. Also, every augmenting path edge can be shown to not participate in any future augmenting paths that are computed in the current phase. By using these facts, one can show that the edges visited unsuccessfully by a DFS will not lead to an augmenting path in the current phase. In our case, however, admissible cycles can exist. Also, some edges on the augmenting path that have zero weight remain admissible after augmentation and may participate in another augmenting path in the current phase.
 We show, however, that any admissible cycle must be completely inside a piece and cannot span multiple pieces (Lemma~\ref{lem:nocycle}). Using this fact, we show that edges visited unsuccessfully by the DFS that do not lie in an affected piece will not participate in any more augmenting paths (Lemma~\ref{lem:final} and Lemma~\ref{lem:one-path}) in the current phase. Therefore, we can safely delete them.


\subparagraph*{Correctness:} From Invariant (I2), each phase of our algorithm will increase the cardinality of $M$ by at least $1$ and so,  our algorithm terminates with a maximum matching.

\subparagraph*{Efficiency:} 
We use the following notations to bound the efficiency of our algorithm.
Let $\{P_1,\ldots,P_t\}$ be the $t$ augmenting paths computed in the second step of the algorithm. Let  $\K_i$ be the set of affected pieces with respect to the augmenting path $P_i$. Let $M_0$ be the matching at the end of the first step of the algorithm. Let, for $1\le i\le t$,  $M_i = M_{i-1}\oplus P_i$, i.e., $M_i$ is the matching after the $i$th augmentation in the second step of the algorithm.  

The first stage is an execution of Dijkstra's algorithm which takes $\BigO(m+n\log n)$ time. Suppose there are $\lambda$ phases; then the cumulative time taken across all phases for the first stage is $\BigO(\lambda m+\lambda n\log n)$. In the second stage, each edge visited by a DFS is discarded for the remainder of the phase, provided it is not in an affected piece. Since each affected piece has $\BigO(mr/n)$ edges, the total time taken by all the DFS searches across all the $\lambda$ phases is bounded by $\BigO((m+n \log n)\lambda + (mr/n)\sum_{i=1}^t |\K_i|)$. In Lemma~\ref{lem:numberofphases}, we bound $\lambda$ by $\sqrt{w}$ and  $\sum_{i=1}^t |\K_i|$ by $\BigO(w\log w)$. Therefore, the total time taken by the algorithm including the time taken by preprocessing step is $\BigO(m\sqrt{r}+m\sqrt{w}+n\sqrt{w}\log n+\frac{mrw\log w}{n}) = \runningtime$.

\begin{lemma}
\label{lem:netcostslack1} \label{lem:netcostslack}
 For any feasible matching $M, y(\cdot)$ maintained by the algorithm, let $y_{\max}$ be the dual weight of every vertex of $B_F$. For any augmenting path $P$ with respect to $M$ from a free vertex $u \in B_F$ to a free vertex $v \in A_F$, 
 \begin{equation*}
     \dist(P)= y_{\max}+\sum_{(a,b) \in P}s(a,b).
 \end{equation*}
\end{lemma}
\begin{proof}
 The weight of $P$ is

\[\dist(P) =  \sum_{(a,b) \in P} \dist(a,b) 
= \sum_{(a,b) \in P\setminus M} (y(b)-y(a) + s(a,b)) + \sum_{(a,b) \in P\cap M} (y(a) - y(b)).\]
Since every vertex on $P$ except for $u$ and $v$ participates in one edge of $P\cap M$ and one edge of $P \setminus M$, we can write the above equation as
\begin{equation}
\label{eq:distance}
    \dist(P) = y(u) - y(v) + \sum_{(a,b)\in P \setminus M} s(a,b) = y(u) - y(v) + \sum_{(a,b)\in P} s(a,b).
\end{equation}

The last equality follows from the fact that edges of $P\cap M$ satisfy~\eqref{eq:feas2} and have a slack of zero. From (I1), we get that $y(u)=y_{\max}$ and $y(v)=0$, which gives,
 \begin{equation*}
     \dist(P)= y_{\max}+\sum_{(a,b) \in P}s(a,b).
 \end{equation*}
\end{proof}
\begin{lemma}
\label{lem:nocycle}
For any feasible matching $M, y(\cdot)$ maintained by the algorithm, and for any alternating cycle $C$ with respect to $M$, if $\dist(C)>0$, then 
 \begin{equation*}
     \sum_{(a,b) \in P}s(a,b) > 0,
 \end{equation*}
 i.e., $C$ is not a cycle in the admissible graph.
\end{lemma}
\begin{proof}
The claim follows from~\eqref{eq:distance} and the fact that the first vertex $u$ and the last vertex $v$ in a cycle are the same. 
\end{proof}
\begin{lemma}
\label{lem:numberofphases}
The total number of phases is  $\BigO(\sqrt{w})$ and the total number of affected pieces is $\BigO(w\log w)$, i.e., $\sum_{i=1}^t |\K_i| = \BigO(w\log w)$.
\end{lemma}
\begin{proof}
Let $M_{\opt}$ be a maximum matching, which has weight at most $w$. Consider any phase $k$ of the algorithm. By (I2), the dual weight $y_{\max}$ of every free vertex in $B_F$ is at least $k$. The symmetric difference of $M$ and $M_{\opt}$ will contain $j=|M_{\opt}|-|M|$ vertex-disjoint augmenting paths. Let $\{\mathcal{P}_1,\ldots, \mathcal{P}_j\}$ be these augmenting paths. These paths contain edges of $M_\opt$ and $M$, both of which are of weight at most $w$. Therefore, the sum of weights of these paths is
\begin{equation*}
\sum_{i=1}^j \dist(\mathcal{P}_i) \le 2w.
\end{equation*}
Let $y_{\max}$ be the dual weight of every vertex $b$ of $B$ that is free with respect to $M$. i.e., $b \in B_F$. From (I2), $y_{\max} \ge k$. From Lemma~\ref{lem:netcostslack} and the fact that the slack on every edge is non-negative, we immediately get,
\begin{equation}
\label{eq:phaseslength}
2w \ge \sum_{i=1}^j\dist(\mathcal{P}_i) \ge j y_{\max}\ge  jk.
\end{equation}
When $\sqrt{w} \leq k <\sqrt{w}+1$, it follows from the above equation that $j = |M_{\opt}|-|M| \le  2\sqrt{w}$. From (I2), we will compute at least one augmenting path in each phase and so the remaining $j$ unmatched vertices are matched in at most $2\sqrt{w}$ phases. This bounds the total number of phases by $3\sqrt{w}$.

Recollect that $\{P_1,\ldots, P_t\}$ are the augmenting paths computed by the algorithm.  The matching $M_0$ has $|M_{\opt}|-t$ edges. Let $y_{\max}^l$ correspond to the dual weight of the free vertices of $B_F$ when the augmenting path $P_l$ is found by the algorithm. From Lemma~\ref{lem:netcostslack},  and the fact that $P_l$ is an augmenting path consisting of zero slack edges, we have $y_{\max}^l = \dist(P_l)$. Before augmenting along $P_l$, there are $|M_{\opt}|-t+l-1$ edges in $M_{l-1}$ and $j=|M_{\opt}|-|M_{l-1}| = t-l+1$. Plugging this in to~\eqref{eq:phaseslength}, we get $\dist(P_l)= y_{\max}^l \le \frac{2w}{t-l+1}$. Summing over all $1\le l\le t$, we get, 
\begin{equation}
\label{eq:lengthbound}
\sum_{l=1}^t \dist(P_l) \le w \sum_{l=1}^t \frac{2}{t-l+1} = \BigO(w \log t) = \BigO(w \log w).
\end{equation}
 For any augmenting path $P_l$, the number of affected pieces is upper bounded by the number of non-zero weight edges on $P_l$, i.e., $|\K_l| \le \dist(P_l)$. Therefore, 
 \[\sum_{l=1}^t |\K_l| \le \sum_{l=1}^t \dist(P_l) = \BigO(w \log w).\]
\end{proof}

\section{Proof of invariants}
We now prove (I1) and (I2).
 Consider any phase $k$ in the algorithm. Assume inductively that at the end of phase $k-1$, (I1) and (I2) hold. We will show that (I1) and (I2) also hold at the end of the phase $k$.  We establish a lemma that will help us prove (I1) and (I2).
 
 \begin{lemma}
 \label{lem:matchedge}
 For any edge $(a,b) \in M$, let $\ell_a$ and $\ell_b$ be the distances returned by Dijkstra's algorithm during the first stage of phase $k$, then $\ell_a=\ell_b$.
 \end{lemma}
\begin{proof}
 The only edge directed towards $b$ is an edge from its match $a$. Therefore, any path from $s$ to $b$ in the augmented residual network, including the shortest path, should pass through $a$. Since the slack on any edge of $M$ is $0$, $\ell_b=\ell_a+s(a,b)= \ell_a$.   
 \end{proof}
 
\begin{lemma}
\label{lem:feas}
 Any matching $M$ and dual weights $y(\cdot)$ maintained during the execution of the algorithm are feasible. 
\end{lemma}
 \begin{proof}
We begin by showing that the dual weight modifications in the first stage of phase $k$ will not violate dual feasibility conditions~\eqref{eq:feas1} and~\eqref{eq:feas2}. Let $\tilde{y}(\cdot)$ denote the dual weights after the execution of the first stage of the algorithm. Consider any edge $(u,v)$ directed from $u$ to $v$. 
There are the following possibilities:

 If both $\ell_u$ and $\ell_v$ are greater than or equal to $\ell$, then $y(u)$ and $y(v)$ remain unchanged and the edge remains feasible. 

 If both $\ell_u$ and $\ell_v$ are less than $\ell$, suppose $(u,v) \in M$. Then, from Lemma~\ref{lem:matchedge}, $\ell_u=\ell_v$. We have, $\tilde{y}(u) = y(u)+\ell-\ell_u$, $\tilde{y}(v)=y(v)+\ell-\ell_v$, and $\tilde{y}(u)-\tilde{y}(v)= y(u)-y(v) +\ell_v - \ell_u = \dist(u,v)$ implying $(u,v)$ satisfies~\eqref{eq:feas2}.
 
 If $\ell_u$ and $\ell_v$ are less than $\ell$ and $(u,v)\not\in M$, then $u \in B$ and $v \in A$. By definition, $y(u) - y(v) + s(u,v)= \dist(u,v)$. By the properties of shortest paths, for any edge $(u,v)$, $\ell_v -\ell_u \le s(u,v)$.  The dual weight of $u$ is updated to $y(u)+\ell-\ell_u$ and dual weight of $v$ is updated to $y(v)+\ell-\ell_v$. The difference in the updated dual weights $\tilde{y}(u)-\tilde{y}(v) = (y(u) + \ell - \ell_u) - (y(v) + \ell - \ell_v) = y(u) - y(v) + \ell_v-\ell_u \le y(u) - y(v)+s(u,v) = \dist(u,v)$. Therefore, $(u,v)$ satisfies~\eqref{eq:feas1}.

If $\ell_u < \ell$ and $\ell_v \ge \ell$, then, from Lemma~\ref{lem:matchedge}, $(u,v)\not\in M$, and so $u \in B$ and $v \in A$. From the shortest path property, for any edge $(u,v)$, $\ell_v -\ell_u \le s(u,v)$. Therefore,
\begin{equation*}
    \tilde{y}(u) - \tilde{y}(v) =y(u)-y(v)+\ell-\ell_u \le y(u)-y(v)+\ell_v -\ell_u \le y(u)-y(v)+s(u,v) = \dist(u,v),
\end{equation*}
implying $(u,v)$ satisfies~\eqref{eq:feas1}.

If $\ell_u \ge \ell$ and $\ell_v < \ell$, then, from Lemma~\ref{lem:matchedge}, $(u,v)\not\in M$, and so $u \in B$ and $v \in A$. Since $\ell_v < \ell$,  we have,
\begin{equation*}
    \tilde{y}(u) - \tilde{y}(v) =y(u)-y(v)-\ell+\ell_v < y(u)-y(v) \le \dist(u,v),
\end{equation*}
implying $(u,v)$ satisfies~\eqref{eq:feas1}.

In the second stage of the algorithm, when an augmenting path $P$ is found, the dual weights of some vertices of $B$ on $P$ decrease and the directions of edges of $P$ change. We argue these operations do not violate feasibility. Let $\tilde{y}(\cdot)$ be the dual weights after these operations. Consider any edge $(a,b) \in A \times B$.  If $b$ is not on $P$, then the feasibility of $(a,b)$ is unchanged. If $b$ is on $P$ and $a$ is not on $P$, then $\tilde{y}(b) \leq y(b)$, and $\tilde{y}(b) - \tilde{y}(a) \leq y(b) - y(a) \leq \dist(a,b)$, implying \eqref{eq:feas1} holds. The remaining case is when both $a$ and $b$ are on $P$. Consider if $(a,b) \in M$ after augmentation. Prior to augmentation, $(a,b)$ was an admissible edge not in $M$, and we have $y(b) - y(a) = \dist(a,b)$ and $\tilde{y}(b) = y(b) - 2c(a,b)$. So, $\tilde{y}(a) - \tilde{y}(b) = y(a) - (y(b) - 2\dist(a,b)) = y(a) - y(b) + 2\dist(a,b) = \dist(a,b)$, implying \eqref{eq:feas2} holds. Finally, consider if $(a,b) \notin M$ after augmentation. Then, prior to augmentation, $(a,b)$ was in $M$, and $y(a) - y(b) = \dist(a,b)$. So, $\tilde{y}(b) - \tilde{y}(a) \leq y(b) - y(a) = -\dist(a,b) \leq \dist(a,b)$, implying \eqref{eq:feas1} holds. We conclude the second stage maintains feasibility.

\end{proof}
Next we show that the dual weights $A_F$ are zero and the dual weights of all vertices of $B_F$ are equal to $\ymax$. At the start of the second step, all dual weights are $0$. During the first stage, the dual weight of any vertex $v$ will increase by $\ell - \ell_v$ only if $\ell_v < \ell$. By~\eqref{eq:closestfreevertex}, for every free vertex $a \in A_F$, $\ell_a \ge \ell$, and so the dual weight of every free vertex of $A$ remains unchanged at $0$. Similarly, for any free vertex $b \in B_F$, $\ell_b = 0$, and the dual weight increases by $\ell$, which is the largest possible increase. This implies that every free vertex in $B_F$ will have the same dual weight of $y_{\max}$. In the second stage, matched vertices of $B$ undergo a decrease in their dual weights, which does not affect vertices in $B_F$. Therefore, the dual weights of vertices of $B_F$ will still have a dual weight of $y_{\max}$ after stage two. This completes the proof of (I1).

Before we prove (I2), we will first establish a property of the admissible graph after the dual weight modifications in the first stage of the algorithm.
 \begin{lemma}
 \label{lem:firststage}
 After the first stage of each phase, there is an augmenting path consisting of admissible edges.
 \end{lemma}
  \begin{proof}
  Let $a \in A_F$ be a free vertex whose shortest path distance from $s$ in the augmented residual network is $\ell$, i.e., $\ell_a=\ell$. Let $P$ be the shortest path from $s$ to $a$ and let $P_a$ be the path $P$ with $s$ removed from it. Note that $P_a$ is an augmenting path. We will show that after the dual updates in the first stage, every edge of $P_a$ is admissible.
Consider any edge $(u,v) \in P_a \cap M$, where $u \in A$ and $v \in B$. From Lemma~\ref{lem:matchedge}, $\ell_u=\ell_v$. Then the updated dual weights are $\tilde{y}(u) = y(u)+\ell-\ell_u$ and $\tilde{y}(v) = y(v)+\ell-\ell_v$.   Therefore, $\tilde{y}(u)-\tilde{y}(v)=y(u)-y(v) -\ell_u+\ell_v= \dist(u,v)$, and $(u,v)$ is admissible.
Otherwise, consider any edge $(u,v) \in P_a \setminus M$, where $u \in B$ and $v\in A$. From the optimal substructure property of shortest paths, for any edge $(u,v) \in P_a$ directed from $u$ to $v$,  $\ell_v-\ell_u=s(u,v)$. Therefore, the difference of the new dual weights is
\begin{equation*}
    \tilde{y}(u)-\tilde{y}(v)= y(u)+\ell-\ell_u -y(v) -\ell+\ell_v = y(u)-y(v)-\ell_u+\ell_v = y(u)-y(v)+s(u,v)=\dist(u,v),
\end{equation*}
implying that $(u,v)$ is admissible.
 \end{proof}

\subparagraph*{Proof of (I2):} From Lemma~\ref{lem:firststage}, there is an augmenting path of admissible edges at the end of the first stage of any phase. Since we execute a DFS from every free vertex $b \in B_F$ in the second stage, we are guaranteed to find an augmenting path. Next, we show in Corollary~\ref{cor:one-path} that there is no augmenting path of admissible edges at the end of stage two of phase $k$, i.e., all augmenting paths in the residual network have a slack of at least $1$. This will immediately imply that the first stage of phase $k+1$ will have to increase the dual weight of every free vertex by at least $1$ completing the proof for (I2). 

Edges that are deleted during a phase do not participate in any augmenting path for the rest of the phase. We show this in two steps. First, we show that at the time of deletion of an edge $(u,v)$, there is no path in the admissible graph that starts from the edge $(u,v)$ and ends at a free vertex $a \in A_F$ (Lemma~\ref{lem:one-path}). In Lemma~\ref{lem:final}, we show that any such edge $(u,v)$ will not participate in any admissible alternating path to a free vertex of $A_F$ for the rest of the phase. 

We use DFS$(b,k)$ to denote the DFS initiated from $b$ in phase $k$. Let $P^b_u$ denote the path maintained by DFS$(b,k)$ when the vertex $u$ was added to the path. 

\begin{lemma}
\label{lem:final}
Consider some point during the second stage of phase $k$ where there is an edge $(u,v)$ that does not participate in any admissible alternating path to a vertex of $A_F$. Then, for the remainder of phase $k$, $(u,v)$ does not participate in any admissible alternating path to a vertex of $A_F$.
\end{lemma}
\begin{proof}
Assume for the sake of contradiction that at some later time during phase $k$, $(u,v)$ becomes part of an admissible path $P_{y,z}$ from a vertex $y$ to a vertex $z \in A_F$. Consider the first time this occurs for $(u,v)$. 
During the second stage, the dual weights of some vertices of $B$ may decrease just prior to augmentation; however, this does not create any new admissible edges. Therefore, $P_{y,z}$ must have become an admissible path due to augmentation along a path $P_{a,b}$ from some $b \in B_F$ to some $a \in A_F$.
Specifically, $P_{y,z}$ must intersect $P_{a,b}$ at some vertex $x$. Therefore, prior to augmenting along $P_{a,b}$, there was an admissible path from $y$ to $a$ via $x$. This contradicts the assumption that $(u,v)$ did not participate in any admissible path to a vertex of $A_F$ prior to this time.
\end{proof}

\begin{lemma}
\label{lem:dfsprop}
Consider the execution of DFS$(b,k)$ and the path $P_u^b$. Suppose the DFS$(b,k)$ marks an edge $(u,v)$ as visited. Let $P_{v}$ be an admissible alternating path from $v$ to any free vertex $a \in A_F$ in $G'$. Suppose $P_{v}$ and $P_u^b$ are vertex-disjoint. Then, DFS$(b,k)$ will find an augmenting path that includes the edge $(u,v)$.  
\end{lemma}
\begin{proof}
$P_v$ and $P^b_u$ are vertex-disjoint and so, $v$ is not on the path $P^b_u$. Therefore, DFS($b,k$) will add $(u,v)$ to the path and we get the path $P=P^b_v$.  We will show that all edges of $P_v$ are unvisited by DFS($b,k$), and so the DFS procedure, when continued from $v$, will discover an augmenting path.

We show, through a contradiction, that all edges of $P_v$ are not yet visited by DFS($b,k$). Consider, for the sake of contradiction, among all the edges of $P_v$, the edge $(u',v')$ that was marked visited first. We claim the following:
\begin{itemize}
    \item[(i)] {\it $(u',v')$ is visited before $(u,v)$}: This follows from the assumption that when $(u,v)$ was marked as visited, $(u',v')$ was already marked as visited by the DFS. 
    \item[(ii)] {\it $(u,v)$ is not a descendant of $(u',v')$ in the DFS}:
    If $(u',v')$ was an ancestor of $(u,v)$ in the DFS, then $P^b_u$ contains $(u',v')$. By definition, $P_v$ also contains $(u',v')$, which contradicts the assumption that $P^b_u$ and $P_v$ are disjoint paths.
    \item[(iii)] {\it When $(u',v')$ is marked visited, it will be added to the path by the DFS}: The only reason why $(u',v')$ is visited but not added is if $v'$ is already on the path $P^b_{u'}$. In that case, $P_v$ and $P^b_{u'}$ will share an edge that was visited before $(u',v')$ contradicting the assumption that $(u',v')$ was the earliest edge of $P_v$ to be marked visited.   
\end{itemize}
From (iii), when $(u',v')$ was visited, it was added to the path $P^b_{v'}$. Since $(u',v')$ was the edge on $P_v$ that was marked visited first by DFS($b,k$), all edges on the subpath from $v'$ to $a$ are unvisited. Therefore, the DFS($b,k$), when continued from $v'$, will not visit $(u,v)$ (from (ii)), will find an augmenting path, and terminate. From (i), $(u,v)$ will not be marked visited by DFS($b,k$) leading to a contradiction.  
\end{proof}

 \begin{lemma}
 \label{lem:one-path}
Consider a DFS initiated from some free vertex $b\in B_F$ in phase $k$. Let $M$ be the matching at the start of this DFS and $M'$ be the matching when the DFS terminates. Suppose the edge $(u,v)$ was deleted during DFS$(b,k)$. Then there is no admissible path starting with $(u,v)$  and ending at a free vertex $a \in A_F$ in $G_{M'}$.
 \end{lemma}
\begin{proof}
At the start of phase $k$, $G'$ is initialized to the admissible graph. Inductively, we assume that all the edges discarded in phase $k$ prior to the execution of DFS($b,k$) do not participate in any augmenting path of admissible edges with respect to $M$. Therefore, any augmenting path of admissible edges in $G_M$ remains an augmenting path in $G'$. There are two possible outcomes for DFS($b,k$). Either, (i) the DFS terminates without finding an augmenting path, or (ii) the DFS terminates with an augmenting path $\tilde{P}$ and $M'=M\oplus \tilde{P}$.

In case (i), $M=M'$ and any edge $(u,v)$ visited by the DFS($b,k)$ is marked for deletion. For the sake of contradiction, let $(u,v)$  participate in an admissible path $P$ to a free vertex $a' \in A_F$.  Since $u$ is reachable from $b$ and $a'$ is reachable from $u$ in $G_{M}$, $a'$ is reachable from $b$. This contradicts the fact that DFS($b,k$) did not find an augmenting path. Therefore, no edge $(u,v)$ marked for deletion participates in an augmenting path with respect to $M$.

In case (ii), $M'=M\oplus \tilde{P}$. DFS($b,k$) marks two kinds of edges for deletion. 
\begin{itemize}
\item[(a)] Any edge $(u,v)$ on the augmenting path $\tilde{P}$ such that $\dist(u,v)=1$ is deleted, and,
\item[(b)] Any edge $(u,v)$ that is marked visited by DFS($b,k$), does not lie on $\tilde{P}$, and does not belong to any affected piece is deleted. 
\end{itemize}
In (a), there are two possibilities (1) $(u,v) \in \tilde{P} \cap M$ or (2) $(u,v)\in \tilde{P}\setminus M$. If $(u,v) \in M$ (case (a)(1)), then, after augmentation along $\tilde{P}$, $s(u,v)$ increases from $0$ to at least $2$, and $(u,v)$ is no longer admissible. Therefore, $(u,v)$ does not participate in any admissible alternating paths to a free vertex in $A_F$ with respect to $G_{M'}$. If $(u,v) \not\in M$ (case (a)(2)), then the \augment\ procedure reduces the dual weight of $u \in B$ by $2$. So, every edge going out of $u$ will have a slack of at least $2$. Therefore, $(u,v)$ cannot participate in any admissible path $P$ to a free vertex in $A_F$. This completes case (a).

For (b), we will show that $(u,v)$, even prior to augmentation along $\tilde{P}$, did not participate in any path of admissible edges from $v$ to any free vertex of $A_F$. For the sake of contradiction, let there be a path $P_v$ from $v$ to $a' \in A_F$. We claim that $P_v$ and $P^b_u$ are not vertex-disjoint. Otherwise, from Lemma~\ref{lem:dfsprop}, the path $\tilde{P}$ found by DFS($b,k$) includes $(u,v)$. However, by our assumption for case (b), $(u,v)$ does not lie on $\tilde{P}$. Therefore, we safely assume that $P_v$ intersects $P^b_u$. There are two cases:
\begin{itemize}
    \item {\it $\dist(u,v)=1$}: We will construct a cycle of admissible edges containing the edge $(u,v)$. Since $\dist(u,v)=1$, our construction will contradict Lemma~\ref{lem:nocycle}. Let $x$ be the first vertex common to both $P_v$ and $P^b_u$ as we walk from $v$ to $a'$ on $P_v$. To create the cycle, we traverse from $x$ to $u$ along the path $P^b_u$, followed by the edge $(u,v)$, followed by the path from $v$ to $x$ along $P_v$. All edges of this cycle are admissible including the edge $(u,v)$.
    \item {\it $\dist(u,v)=0$}: In this case, $(u,v)$ belongs to some piece $K_i$ that is not an affected piece. Among all edges visited by DFS($b,k$), consider the edge $(u',v')$ of $K_i$, the same piece as $(u,v)$, such that $v'$ has a path to the vertex $a'\in A_F$ with the fewest number of edges. Let $P_{v'}$ be this path. We claim that $P_{v'}$ and $P^b_{u'}$ are not vertex-disjoint. Otherwise, from Lemma~\ref{lem:dfsprop}, the path $\tilde{P}$ found by DFS($b,k$)  includes $(u',v')$ and $K_i$ would have been an affected piece. Therefore, we can safely assume that $P_{v'}$ intersects with $P^b_{u'}$. Let $z$ be the first intersection point with $P^b_{u'}$ as we walk from $v'$ to $a'$ and let $z'$ be the vertex that follows after $z$ in $P^b_{u'}$. There are two possibilities:
    \begin{itemize}
        \item {\it The edge $(z,z')\in K_i$:} In this case, $(z,z')$ is also marked visited by DFS($b,k$), and $z'$ has path to $a'$ with fewer number of edges than $v'$. This contradicts our assumption about $(u',v')$.
        \item {\it The edge $(z,z') \not\in K_i$:} In this case, consider the cycle obtained by walking from $z$ to $u'$ along the path $P^b_{u'}$ followed by the edge $(u',v')$ and the path from $v'$ to $z$ along $P_{v'}$. Since $(u',v') \in K_i$ and $(z,z') \not\in K_i$, the admissible cycle contains at least one edge of weight $1$. This contradicts Lemma~\ref{lem:nocycle}. 
    \end{itemize} 
    \end{itemize}
    This concludes case (b) which shows that $(u,v)$ did not participate in any augmenting paths with respect to $M$. From Lemma~\ref{lem:final}, it follows that $(u,v)$ does not participate in any augmenting path with respect to $G_{M'}$ as well.
 \end{proof}

 \begin{corollary}
 \label{cor:one-path}
 At the end of any phase, there is no augmenting path of admissible edges. 
 \end{corollary}
 
\section{Minimum bottleneck matching}
\label{sec:bottleneck}
We are given two sets $A$ and $B$ of $n$ $d$-dimensional points. Consider a weighted and complete bipartite graph on points of $A$ and $B$. The weight of any edge $(a,b) \in A\times B$ is given by its Euclidean distance and denoted by $\|a-b\|$.  For any matching $M$ of $A$ and $B$ let its largest weight edge be its \emph{bottleneck edge}. In the \emph{minimum bottleneck matching} problem, we wish to compute a matching $M_\opt$ of $A$ and $B$ with the smallest weight bottleneck edge. We refer to this weight as the \emph{bottleneck distance} of $A$ and $B$ and denote it by $\beta^*$.  An \emph{$\eps$-approximate bottleneck matching} of $A$ and $B$ is any matching $M$  with a bottleneck edge weight of at most $(1+\eps)\beta^*$. 
We present an algorithm that takes as input $A,B$, and a value $\delta$  such that $\beta^*\le \delta \le  (1+\eps/3)\beta^*$, and produces an $\eps$-approximate  bottleneck matching. For simplicity in presentation, we describe our algorithm for the $2$-dimensional case when all points of $A$ and $B$ are in a bounding square $S$. The algorithm easily extends to any arbitrary fixed dimension $d$. For $2$-dimensional case, given a value $\delta$, our algorithm executes in $\BigOT(n^{4/3}/\eps^3)$ time. 

Although, the value of $\delta$ is not known to the algorithm, we can first find a value $\alpha$ that is guaranteed to be an $n$-approximation of the bottleneck distance~\cite[Lemma 2.2]{av_scg04} and then select $\BigO(\log n/\eps)$ values from the interval $[\alpha/n, \alpha]$ of the form $(1+\eps/3)^i\alpha/n$, for $0\le i \le \BigO(\log n /\eps)$.  We will then execute our algorithm for each of these $\BigO(\log n/\eps)$ selected values of $\delta$.  Our algorithm returns a maximum matching whose edges are of length at most $(1+\eps/3)\delta$ in $\BigO(n^{4/3}/\eps^3)$ time. At least one of the $\delta$ values chosen will be a $\beta^* \le \delta \le (1+\eps/3)\beta^*$. The matching returned by the algorithm for this value of $\delta$ will be perfect ($|M|=n$) and have a bottleneck edge of weight at most  $(1+\eps/3)^2\beta^* \le (1+\eps)\beta^*$ as desired.  Among all executions of our algorithm that return a perfect matching, we return a perfect matching with the smallest bottleneck edge weight. Therefore, the total time taken to compute the $\eps$-approximate bottleneck matching is $\BigOT(n^{4/3}/\eps^4)$. 

Given the value of $\delta$, the algorithm will construct a graph as follows: Let $\grid$ be a grid on the bounding square $S$. The side-length of every square in this grid is $\eps\delta/(6\sqrt{2})$. For any cell $\xi$ in the grid $\grid$, let $N(\xi)$ denote the subset of all cells $\xi'$ of $\grid$ such that the minimum distance between $\xi$ and $\xi'$ is at most $\delta$. By the use of a simple packing argument, it can be shown that $|N(\xi)| = \BigO(1/\eps^2)$. 

For any point $v \in A\cup B$, let $\xi_v$ be the cell of grid $\grid$ that contains $v$. We say that a cell $\xi$ is \emph{active} if $(A\cup B)\cap \xi \neq \emptyset$. Let $A_\xi$ and $B_\xi$ denote the points of $A$ and $B$ in the cell $\xi$. We construct a bipartite graph $\mathcal{G}(A\cup B, \mathcal{E})$ on the points in $A\cup B$ as follows: For any pair of points $(a,b) \in A\times B$, we add an edge in the graph if $\xi_b \in N(\xi_a)$. Note that every edge $(a,b)$ with $\|a-b\| \le \delta$ will be included in $\mathcal{G}$.  Since $\delta$ is at least the bottleneck distance, $\mathcal{G}$ will have a perfect matching.  The maximum distance between any cell $\xi$ and a cell in $N(\xi)$ is $(1+\eps/3)\delta$. Therefore, no edge in $\mathcal{G}$ will have a length greater than $(1+\eps/3)\delta$. This implies that any perfect matching in $\mathcal{G}$ will  also be an $\eps$-approximate bottleneck matching. We use our algorithm for maximum matching to compute this perfect matching in $\mathcal{G}$. Note, that $\mathcal{G}$ can have $\Omega(n^2)$ edges. For the sake of efficiency, our algorithm executes on a compact representation of $\mathcal{G}$ that is described later. Next, we assign weights of $0$ and $1$ to the edges of $\mathcal{G}$ so that the any maximum matching in $\mathcal{G}$ has a small weight $w$.

For a parameter\footnote{Assume $r$ to be a perfect square.} $r > 1$, we will carefully select another grid $\grid'$ on the bounding square $S$, each cell of which has a side-length of $\sqrt{r}(\eps\delta/(6\sqrt{2}))$ and encloses $\sqrt{r}\times \sqrt{r}$ cells of $\grid$. For any cell $\xi$ of the grid $\grid$, let $\Box_{\xi}$ be the cell in $\grid'$ that contains $\xi$.  Any cell $\xi$ of $\grid$ is a \emph{boundary cell} with respect to $\grid'$ if there is a cell $\xi' \in N(\xi)$ such that $\Box_{\xi'}\neq \Box_{\xi}$. Equivalently, if the minimum distance from $\xi$ to $\Box_{\xi}$ is at most $\delta$, then $\xi$ is a boundary cell.  For any boundary cell $\xi$ of $\grid$ with respect to grid $\grid'$, we refer to all points of $A_\xi$ and $B_\xi$ that lie in $\xi$ as boundary points. All other points of $A$ and $B$ are referred to as internal points.  We carefully construct this grid $\grid'$ such that the total number of boundary points is $\BigO(n/\eps\sqrt{r})$ as follows: First, we will generate the vertical lines for  $\grid'$, and then we will generate the horizontal lines using a similar construction.  Consider the vertical line $y_{ij}$ to be the line $x=i(\eps\delta)/(6\sqrt{2})+j\sqrt{r}(\eps\delta/(6\sqrt{2}))$. For any fixed integer $i$ in $[1, \sqrt{r}]$, consider the set of vertical lines $\mathbb{Y}_i=\{y_{ij}\mid y_{ij} \textnormal{ intersects the bounding square } S\}$. We label all cells $\xi$ of $\grid$ as boundary cells with respect to $\mathbb{Y}_i$ if the distance from $\xi$ to some vertical line in $\mathbb{Y}_i$ is at most $\delta$. We designate the points inside the boundary cells as boundary vertices with respect to $\mathbb{Y}_i$. For any given $i$, let $A_i$ and $B_i$ be the boundary vertices of $A$ and $B$ with respect to the lines in $\mathbb{Y}_i$. We select an integer $\kappa = \arg\min_{1\le i\le \sqrt{r}} |A_i\cup B_i|$ and use $\mathbb{Y}_{\kappa}$ as the vertical lines for our grid $\grid'$. We use a symmetric construction for the horizontal lines.

\begin{lemma}
\label{lem:boundarysize}
Let $A_i$ and $B_i$ be the boundary points with respect to the vertical lines $\mathbb{Y}_i$. Let $\kappa = \arg\min_{1\le i\le \sqrt{r}} |A_i\cup B_i|$. Then, $|A_\kappa\cup B_\kappa| = \BigO(n/(\eps\sqrt{r}))$.
\end{lemma}
\begin{proof}
For any fixed cell $\xi$ in $\grid$, of the $\sqrt{r}$ values of $i$, there are $\BigO(1/\eps)$ values for which $\mathbb{Y}_i$ has a vertical line at a distance at most $\delta$ from $\xi$. Therefore, each cell $\xi$ will be a boundary cell in only $\BigO(1/\eps)$ shifts out of $\sqrt{r}$ shifts. So, $A_\xi$ and $B_\xi$ will be counted in $A_i \cup B_i$ for $\BigO(1/\eps)$ different values of $i$. Therefore, if we take the average over choices of $i$, we get
\begin{equation*}
    \min_{1\le i\le \sqrt{r}}|A_i\cup B_i| \le \frac{1}{\sqrt{r}}\sum_{i=1}^{\sqrt{r}} |A_i\cup B_i| \le \BigO(n/(\eps\sqrt{r})).
\end{equation*}
\end{proof}
Using a similar construction, we guarantee that the boundary points with respect to the horizontal lines of $\grid'$ is also at most $\BigO(n/(\eps\sqrt{r}))$. 
\begin{corollary}

\label{cor:smallweight}
The grid $\grid'$ that we construct has $\BigO(n/(\eps\sqrt{r}))$ many boundary points.
\end{corollary}

For any two cells $\xi$ and $\xi' \in N(\xi)$ of the grid $\grid$, suppose $\Box_{\xi}\neq \Box_{\xi'}$. Then the weights of all edges of $A_{\xi}\times B_{\xi'}$ and of $B_\xi\times A_{\xi'}$ are set to $1$. All other edges have a weight of $0$. We do not make an explicit weight assignment as it is expensive to do so. Instead, we can always derive the weight of an edge when we access it. Only boundary points will have edges of weight $1$ incident on them. From Corollary~\ref{cor:smallweight}, it follows that any maximum matching will have a weight of $w = \BigO(n/(\eps\sqrt{r}))$.  

The edges of every piece in $\mathcal{G}$ have endpoints that are completely inside a cell of $\grid'$. Note, however, that there is no straight-forward bound on the number of points and edges of $\mathcal{G}$ inside each piece. Moreover, the number of edges in $\mathcal{G}$ can be $\Theta(n^2)$. Consider any feasible matching $M, y(\cdot)$ in $\mathcal{G}$. Let $\mathcal{G}_M$ be the residual network. In order to obtain a running time of $\BigOT(n^{4/3}/\eps^3)$, we use the grid $\grid$ to construct a compact residual network $\mathcal{CG}_M$ for any feasible matching $M,y(\cdot)$ and use this compact graph to implement our algorithm. The following lemma assists us in constructing the compressed residual network.

\begin{lemma}
\label{lem:differ1}
Consider any feasible matching $M, y(\cdot)$ maintained by our algorithm on $\mathcal{G}$ and  any active cell $\xi$ in the grid $\grid$. The dual weight of any two points $a,a' \in A_{\xi}$ can differ by at most $2$. Similarly, the dual weights of any two points $b,b' \in B_{\xi}$ can differ by at most $2$.
\end{lemma}
\begin{proof}
We present our proof for two points $b,b' \in B_{\xi}$. A similar argument will extend for $a,a'\in A_{\xi}$. For the sake of contradiction, let $y(b) \ge y(b')+3$. $b'$ must be matched since $y(b') < y(b) \le y_{\max}$. Let $m(b') \in A$ be the match of $b'$  in $M$. From~\eqref{eq:feas2}, $y(m(b'))-y(b')=\dist(b',m(b'))$. Since both $b$ and $b'$ are in $\xi$, the distance $\dist(b,m(b'))=\dist(b',m(b'))$. So, $y(b)-y(m(b'))\ge (y(b') +3)-y(m(b'))  = 3-\dist(b,m(b'))$. This violates~\eqref{eq:feas1} leading to a contradiction. 
\end{proof}
For any feasible matching and any cell $\xi$ of $\grid$, we divide points of $A_{\xi}$ and $B_{\xi}$ based on their dual weight into at most three clusters. Let $A_{\xi}^1, A_{\xi}^2$ and $A_{\xi}^3$ be the three clusters of points in $A_{\xi}$ and let $B_{\xi}^1, B_{\xi}^2$ and $B_{\xi}^3$ be the three clusters of points in $B_{\xi}$. We assume that points with the largest dual weights are in $A_{\xi}^1$ (resp. $B_{\xi}^1$), the points with the second largest dual weights are in $A_{\xi}^2$ (resp. $B_{\xi}^2$), and the points with the smallest dual weights are in $A_{\xi}^3$ (resp.  $B_{\xi}^3$). 

\subparagraph*{Compact residual network:} Given a feasible matching $M$, we construct a compact residual network $\mathcal{CG}_M$ to assist in the fast implementation of our algorithm. This vertex set $\mathcal{A}\cup \mathcal{B}$ for the compact residual network is constructed as follows. First we describe the vertex set $\mathcal{A}$. For every active cell $\xi$ in $\grid$, we add a vertex $a_{\xi}^1$ (resp. $a_{\xi}^2, a_{\xi}^3$) to represent the set $A_{\xi}^1$ (resp. $A_{\xi}^2, A_{\xi}^3$) provided $A_{\xi}^1\neq \emptyset$ (resp. $A_{\xi}^2\neq \emptyset, A_{\xi}^3\neq \emptyset$). We designate $a_{\xi}^1$ (resp. $a_{\xi}^2, a_{\xi}^3$) as a \emph{free} vertex if $A_{\xi}^1\cap A_F \neq \emptyset$ (resp.  $A_{\xi}^2\cap A_F\neq \emptyset, A_{\xi}^3\cap A_F\neq \emptyset$). Similarly, we construct a vertex set $\mathcal{B}$ by adding a vertex $b_{\xi}^1$ (resp. $b_{\xi}^2, b_{\xi}^3$) to represent the set $B_{\xi}^1$ (resp. $B_{\xi}^2, B_{\xi}^3$) provided $B_{\xi}^1\neq \emptyset$ (resp. $B_{\xi}^2\neq \emptyset, B_{\xi}^3\neq \emptyset$). We designate $b_{\xi}^1$ (resp. $b_{\xi}^2, b_{\xi}^3$) as a \emph{free} vertex if $B_{\xi}^1\cap B_F \neq \emptyset$ (resp.  $B_{\xi}^2\cap B_F\neq \emptyset, B_{\xi}^3\cap B_F\neq \emptyset$).  Each active cell $\xi$ of the grid $\grid$ therefore has at most six points. Each point in $\mathcal{A}\cup \mathcal{B}$ will  inherit the dual weights of the points in its cluster; for any vertex $a_\xi^1 \in \mathcal{A}$ (resp. $a_\xi^2 \in \mathcal{A}, a_\xi^3 \in \mathcal{A}$), let $y(a_{\xi}^1)$(resp. $y(a_{\xi}^2),y(a_{\xi}^3)$) be the dual weight of all points in $A_{\xi}^1$ (resp. $A_{\xi}^2, A_{\xi}^3$). We define $y(b_{\xi}^1)$, $y(b_{\xi}^2)$, and $y(b_{\xi}^3)$ as dual weights of points in $B_{\xi}^1, B_{\xi}^2$, and $B_{\xi}^3$ respectively. Since there are at most $n$ active cells,  $|\mathcal{A}\cup \mathcal{B}|=\BigO(n)$.  

Next, we create the edge set for the compact residual network $\mathcal{CG}$. For any active cell $\xi$ in the grid $\grid$ and for any cell $\xi'\in N(\xi)$,
\begin{itemize}
    \item We add a directed edge from $a_{\xi}^i$ to $b_{\xi'}^j$, for $i,j \in \{1,2,3\}$ if there is an edge $(a,b) \in (A_{\xi}^i\times B_{\xi'}^j)\cap M$. We define the weight of $(a_{\xi}^i,b_{\xi'}^j)$ to be $\dist(a,b)$. We also define the slack $s(a_{\xi}^i,b_{\xi'}^j)$ to be $\dist(a_{\xi}^i,b_{\xi'}^j) -y(a_{\xi}^i)+y(b_{\xi'}^j)$ which is equal to $ s(a_{\xi}^i,b_{\xi'}^j)=\dist(a,b)-y(a)+y(b) = s(a,b) = 0$.
    \item We add a directed edge from $b_{\xi}^i$ to $a_{\xi'}^j$, for $i,j \in \{1,2,3\}$ if $(B_{\xi}^i \times A_{\xi'}^j)\setminus M \neq \emptyset$. Note that the weight and slack of every directed edge in $B_{\xi}^i \times A_{\xi'}^j$ are identical. We define the weight of $(b_{\xi}^i,a_{\xi'}^j)$ to be $\dist(a,b)$ for any $(a,b) \in A_{\xi'}^j \times B_{\xi}^i$. We also define the slack $s(b_{\xi}^i,a_{\xi'}^j) = \dist(b_{\xi}^i,a_{\xi'}^j) -y(b_{\xi}^i)+y(a_{\xi'}^j)$ which is equal to the slack $s(a,b)$.
\end{itemize}

For each vertex in $\mathcal{A}\cup \mathcal{B}$, we added at most two edges to every cell $\xi' \in N(\xi)$. Since $N(\xi) = \BigO(1/\eps^2)$, the total number of edges in $\mathcal{E}$ is $\BigO(n/\eps^2)$. For a cell $\Box$ in $\grid'$, let $\mathcal{A}_{\Box}$ be the points of $\mathcal{A}$ generated by cells of $\grid$ that are contained inside the cell $\Box$. A piece $K_\Box$ has $\mathcal{A}_{\Box} \cup \mathcal{B}_{\Box}$ as the vertex set and $\mathcal{E}_{\Box}=((\mathcal{A}_{\Box}\times\mathcal{B}_{\Box})\cup (\mathcal{B}_{\Box}\times \mathcal{A}_{\Box})\cap \mathcal{E})$ as the edge set.  Note that the number of vertices in any piece $K_{\Box}$ is $\BigO(r)$ and the number of edges in $K_{\Box}$ is $\BigO(r/\eps^2)$. Every edge $(u,v)$ of any piece $K_{\Box}$ has a weight $\dist(u,v)=0$ and every edge $(u,v)$ with a weight of zero belongs to some piece of $\mathcal{CG}$.

The following lemma shows that the compact graph $\mathcal{CG}$ preserves all minimum slack paths in $\mathcal{G}_M$.  
\begin{lemma}
\label{lem:compactprop}
For any directed path $\mathcal{P}$ in the compact residual network $\mathcal{CG}$, there is a directed path $P$ in the residual network such that $\sum_{(u,v)\in P}s(u,v) = \sum_{(u,v)\in \mathcal{P}} s(u,v)$.
For any directed path $P$ in $\mathcal{G}_M$, there is a directed path $\mathcal{P}$ in the compact residual network such that
$\sum_{(u,v) \in P} s(u,v) \ge \sum_{(u,v)\in\mathcal{P}}s(u,v).$
\end{lemma}

\subparagraph*{Preprocessing step:} At the start, $M= \emptyset$ and all dual weights are $0$. Consider any cell $\Box$ of the grid $\grid'$ and any cell $\xi$ of $\grid$ that is contained inside $\Box$. Suppose we have a point $a_{\xi}^1$. We assign a demand $d_{a_\xi^1}=|A_{\xi}^1|=|A_{\xi}|$ to $a_{\xi}^1$. Similarly, suppose we have a point $b_{\xi}^1$, we assign a supply $ s_{b_\xi^1}=|B_{\xi}^1|=|B_{\xi}|$. The preprocessing step reduces to finding a maximum matching of supplies to demand. This is an instance of the unweighted transportation problem which can be solved using the algorithm of~\cite{sidford} in $\BigOT(|\mathcal{E}_{\Box}|\sqrt{|\mathcal{A}_{\Box}\cup \mathcal{B}_{\Box}|}) = \BigOT(|\mathcal{E}_{\Box}|\sqrt{r})$.  Every edge of $\mathcal{E}$ participates in at most one piece. Therefore, the total time taken for preprocessing across all pieces is $\BigOT(|\mathcal{E}|\sqrt{r})=\BigOT(n\sqrt{r}/\eps^2)$. We can trivially convert the matching of supplies to demand to a matching in $\mathcal{G}$.

\subparagraph*{Efficient implementation of the second step:} Recollect that the second step of the algorithm consists of phases. Each phase has two stages. In the first stage, we execute Dijkstra's algorithm in $\BigO(n\log n/\eps^2)$ time by using the compact residual network $\mathcal{CG}$. After adjusting the dual weight of nodes in the compact graph, in the second stage, we iteratively compute augmenting paths of admissible edges by conducting a DFS from each vertex. Our implemnetation of DFS has the following differences from the one described in Section~\ref{sec:graphmatch}.
\begin{itemize}
    \item Recollect that each free vertex $v \in \mathcal{B}$ may represent a cluster that has $t>0$ free vertices. We will execute DFS from $v$ exactly $t$ times, once for each of the free vertices of $\mathcal{B}$.
    \item During the execution of any DFS, unlike the algorithm described in Section~\ref{sec:graphmatch}, the DFS will mark an edge as visited only when it backtracks from the edge. Due to this change, all edges on the path maintained by the DFS are marked as unvisited. Therefore,   unlike the algorithm from Section~\ref{sec:graphmatch}, this algorithm will not discard weight $1$ edges of an augmenting path after augmentation. From Lemma~\ref{lem:numberofphases}, the total number of these edges is $\BigO(w\log w)$. 
\end{itemize}

\subparagraph*{Efficiency:} The first stage is an execution of Dijkstra's algorithm which takes $\BigO(|\mathcal{E}|+|\mathcal{V}|\log |\mathcal{V}|)=\BigO(n\log n/\eps^2)$ time. Suppose there are $\lambda$ phases; then the cumulative time taken across all phases for the first stage is $\BigOT(\lambda n/\eps^2)$. In the second stage of the algorithm,  in each phase, every edge is discarded once it is visited by a DFS, unless it is in an affected piece or it is an edge of weight $1$ on an augmenting path. Since each affected piece has $\BigO(r/\eps^2)$ edges, and since there are $\BigO(w \log w)$ edges of weight $1$ on the computed augmenting paths, the total time taken by all the DFS searches across all the $\lambda$ phases is bounded by $\BigOT(n\lambda/\eps^2 + r/\eps^2\sum_{i=1}^t |\K_i| +w\log w)$. In Lemma~\ref{lem:numberofphases}, we bound $\lambda$ by $\sqrt{w}$ and  $\sum_{i=1}^t |\K_i|$ by $\BigO(w\log w)$. Therefore, the total time taken by the algorithm including the time taken by preprocessing step is $\BigOT((n/\eps^2)(\sqrt{r}+\sqrt{w}+\frac{wr}{n}))$. Setting $r=n^{2/3}$, we get $w=\BigO(n/(\eps\sqrt{r})) = \BigO(n^{2/3}/\eps)$, and the total running time of our algorithm is $\BigOT(n^{4/3}/\eps^3)$. To obtain the bottleneck matching, we execute this algorithm on $\BigO(\log (n/\eps))$ guesses; therefore, the total time taken to compute an $\eps$-approximate bottleneck matching is $\BigOT(n^{4/3}/\eps^4)$.
For $d > 2$, we choose $r=n^{\frac{d}{2d-1}}$ and $w=\BigO(n/(d\eps r^{1/d}))$. With these values, the execution time of our algorithm is $\frac{1}{\eps^{\BigO(d)}}n^{1+\frac{d-1}{2d-1}}\mathrm{poly}\log n$.

\end{document}